\documentclass{article}

\usepackage{amsfonts}       

\usepackage{comment}
\usepackage{caption}
\usepackage{subcaption}

\usepackage{amsmath}

\usepackage{amsthm}
\usepackage{amsfonts}
\usepackage{amssymb}
\usepackage{bbold}

\usepackage{mathtools}
\DeclarePairedDelimiter{\ceil}{\lceil}{\rceil}

\newtheorem{thm}{Theorem}[section]
\newtheorem{cor}[thm]{Corollary}
\newtheorem{lem}[thm]{Lemma}

\newcommand{\RR}{ \mathbb{R} }
\newcommand{\NN}{\mathbb{N}}

\newcommand{\spaceo}{\hspace{2 mm}}
\newcommand{\setsep}{ \spaceo | \spaceo}
\newcommand{\half}{\frac{1}{2}}

\newcommand{\Prob}[1]{\mathbb{P}\left( #1 \right)}
\newcommand{\Probu}[2]{\mathbb{P}_{#1}\left( #2 \right)}

\newcommand{\Abs}[1]{\left| #1 \right|}
\newcommand{\Set}[1]{\left\{ #1 \right\}}
\newcommand{\Brack}[1]{\left( #1 \right)}
\newcommand{\BrackSq}[1]{\left[ #1 \right]}
\newcommand{\inner}[2]{\left< #1 , #2 \right>}

\newcommand{\Exp}[1]{ \mathbb{E} #1}
\newcommand{\norm}[1]{\left\|#1\right\|}

\newcommand{\eps}{\varepsilon}
\newcommand{\exponent}[1]{exp\Brack{#1}}

\newcommand{\dbar}[1]{\bar{\bar{#1}}}

\newcommand{\groundX}{\mathbb{X}}

\newcommand{\grmeassingle}{\Delta_{\groundX}}
\newcommand{\grmeas}{\Delta_{\groundX \times \groundX}}
\newcommand{\sqmeas}{\Delta_{S' \times S'}}
\newcommand{\sinmeas}{\Delta_{S'}}

\begin{document}

\title{Clustering Time Series and the Surprising Resilience of HMMs}

\author{Mark Kozdoba, Shie Mannor}
\date{}

\maketitle

\begin{abstract}
Suppose that we are given a time series where consecutive samples are 
believed to come from a probabilistic source, that the source changes
from time to time and that the total number of sources is fixed. 
Our objective is to estimate the distributions of the sources.
A standard approach to this problem is to model the data as a 
hidden Markov model (HMM). However, since the data often lacks the Markov or the 
stationarity properties of an HMM, one can ask whether this approach is still suitable
or perhaps another approach is required. In this paper we show that a maximum likelihood HMM
estimator can be used to approximate the source distributions in a much larger class of models
than HMMs. Specifically, we propose a natural and fairly general non-stationary 
model of the data, where the only restriction is that the sources do not change too often. Our main result 
shows that for this model, a maximum-likelihood HMM estimator produces
the correct second moment of the data, and the results can be extended to higher moments. 
\end{abstract}

\section{Introduction}
Let $x = (x_1,x_2,\ldots,x_N)$ be a sequence of symbols over some 
alphabet $\groundX$, where each symbol is sampled from one of 
$k$ sources, with distributions $\mu_1,\ldots,\mu_k$. Given the sequence $x$, consider the problem of inferring the distribution of the sources, and of 
the classification of the samples, in the sense of determining for each sample $x_i$  which source 
produced it. A well known toy instance of this problem is the ``dishonest 
casino'', where the sources are biased coins, see \cite{Durbin98biologicalsequence}. A classical real world 
application is in speech recognition, see \cite{hmmspeach}. In general, applications appear in virtually any field 
involving time 
series or sequential data. For instance, in financial times series \cite{hmmfinance},  
biological sequence analysis \cite{hmmbiology1},\cite{hmmbiology},\cite{Durbin98biologicalsequence}, computer 
vision \cite{hmmvision}, and climate modelling
\cite{changepoint_climate}. See also \cite{hmmactivity1},\cite{hmmactivity2}, \cite{changepoint_security}, 
\cite{changepoint_deform} for a variety of other applications. Further, the above problem setting can also be 
viewed as a change-point detection problem, \cite{changepointbook} (see also  
\cite{hmm_changepoint} for an HMM based framework), and in particular the applications in 
\cite{hmmfinance},\cite{hmmbiology1},\cite{changepoint_climate},\cite{changepoint_deform}, \cite{regshifts},
 and \cite{hmmactivity2} are of this type. 

When modelling the data $x = (x_1,x_2,\ldots,x_N)$, in order to be able to distinguish between the 
sources, one clearly needs some conditions on how the sources change as time progresses. 
Indeed, if the source is chosen independently at each time $i$, it is 
easy to see that the sources are indistinguishable and one 
effectively sees a single source with distribution equal to the 
empirical distribution of the data. A natural assumption on the 
underlying sequence of sources (also referred to as sequence of states) $s =(s_1,\ldots,s_N)$ is that it 
forms a Markov chain, and the resulting model is a Hidden Markov Model (HMM). In particular this 
assumption is made in all of the above mentioned work. 

With the HMM model, given a sequence of data $x$ one can find a maximum likelihood 
HMM, and then, for instance, use the Viterbi sequence (the most likely state sequence $s$ 
given the data) for classification. However, while the Markov chain assumption 
on the state sequence $s$ is convenient, and there exists a variety of effective inference methods for 
the problem, the data source itself rarely satisfies the Markov condition on the sequence of the 
states. Consider for instance the financial time series applications, as considered in 
\cite{hmmfinance}, \cite{changepoint_deform}, \cite{regshifts}. The data is a time series of stock prices or commodity 
value indices, and the underling hidden states reflect the general conditions of the market, such as 
bull or bear markets. If this data is modelled by an HMM, the model will imply that every day there is 
a certain probability that the market will enter a ``bear'' state, and that the expected time the system 
will spend in this state will be inversely proportional to this probability. Moreover, the model will 
imply that this probability does not change from day to day, due to stationarity of the Markov chain. 
Such properties clearly do not hold for real data as the stock markets are notoriously non-stationary. 
As another example, consider the task of monitoring human physical activity during a day (see, for instance, \cite{hmmactivity2}). Suppose that 
different states of the system correspond to different activities, such as walking, climbing stairs, 
running, driving, riding a bicycle. Assume that time steps are seconds, that at each time instance $i$, 
$x_i$ 
corresponds to some set of features produced by the current activity, and that the activities can be 
distinguished based on the distribution of the features. In this situation, it clearly makes little 
sense to assign probabilities to transitions between, say, walking and climbing stairs states, since 
such probability will depend strongly on the environment, will change on different days and during the day, 
and in any case is likely to be too small to be meaningful. More generally, similar considerations apply in many problem instances where HMMs are used as a change-point detection tool. 

In this paper we show that surprisingly, if one wants to learn the distributions of the sources, 
one can largely ignore the issue of modelling the environment, or modelling the transition mechanism between the 
states, under the assumption that the states do not change too often. 
Specifically, we define an Interval Model $I$ of the 
data $x = (x_1,x_2,\ldots,x_N, \ldots)$ to be a finite or infinite sequence of consecutive intervals in 
$\NN$, $I_1,I_2,\ldots \subset \NN$ with a mapping $\tau:\NN \rightarrow \Set{1,
\ldots,k}$ such that for any $i\in I_l$, $x_i$ has distribution 
$\mu_{\tau(l)}$ and all $x_i$ are independent. As mentioned above, in order to be able to 
differentiate between the sources, one must make \textit{some} assumption on how the source to be 
sampled is chosen at each time instance. 
The assumption that we make in this paper is that for every $l$, $\Abs{I_l} \geq m$ for some  $m > 0$. 
This means that once the system enters a certain state, it stays at least $m$ time units in that state. 
To the best of our knowledge, this assumption did not appear in the literature before.
With this assumption, our main result, Theorem 
\ref{thm:full_statement}, states that if $x$ is a sample from the Interval Model, then a maximum 
likelihood HMM estimator for the sequence $x$ will produce an HMM with source distributions that 
approximate the distributions $\mu_1,\ldots,\mu_k$.  In other words, we show that an HMM estimator 
learns the correct source distributions despite the fact that the sequence was not generated by an HMM. We 
refer to this phenomenon as the \textit{resilience} of the HMM. Our result can be viewed as an 
extension of the classical HMM consistency results, \cite{baum1966}, \cite{petrie} as well as an extension of 
the more recent consistency under misspecification results, \cite{mevel}. On the application side, our results 
provide a better theoretical understanding of the methods that are already widely used. 

The Interval Model with the minimal $m$ duration assumption is a fairly general model. 
Indeed, except for the minimal duration $m$ for each 
state, we \textit{make no other assumptions about the transitions} between the intervals. The transitions 
between different states \textit{need not follow any deterministic or probabilistic pattern}, and in 
particular the process $x$ \textit{does not need to be stationary} and moreover, it \textit{does not need to be 
ergodic}. The intervals themselves also can be of arbitrary lengths, provided it is larger than $m$. 
Therefore the Interval Model setting can be best described as partly stochastic and partly adversarial (see 
\cite{adversarial}). The values
$x_i$ are obtained by sampling from the sources, but the choice of changes of the sources can be arbitrary and hence adversarial. 
In addition, we do not require $m$ to be known apriori but the precision of approximation of Theorem 
\ref{thm:full_statement} will grow with $m$, with explicit bounds.  

While in this paper we are concerned with estimating the source distributions, we note 
that once the sources are known, the problem of classifying the data according to the source is relatively 
easy. For instance one could use a sliding window over the data, $w_i = (x_i,\ldots,x_{i+l})$, and 
for each $i$ decide which source is the most likely to produce $w_i$. Note that if the sources are known, 
one can easily compute the length $l$ of the window that is required to distinguish between the sources with 
high probability. Clearly, the more distinct the sources are, the smaller $l$ is required. In cases 
where $l$ is small compared to $m$, it is straightforward obtain guarantees on the accuracy of this method. 
We note that in contrast, the standard HMM decoding approach, the Viterbi sequence, does not in general 
have any guarantees and for the non-stationary Interval Model type data can be significantly inaccurate.

We now proceed to discuss our results in more detail. Consider a sample $x = (x_1,x_2,
\ldots,x_N)$ generated from an Interval Model $I$ as discussed above. We describe the behaviour of a maximum 
likelihood HMM estimator on such a sequence in two stages. First, we show that with high probability, 
there exists an HMM $H_0$ which assigns a high likelihood to the sequence $x$. Specifically, we show that
there is an HMM that assigns log-likelihood 
\begin{equation}
\label{eq:intro_right_likelihood}
 L(H_0,x) = \frac{1}{N} \log \Probu{H_0}{x} \geq -\frac{\log \Brack{ 2k \cdot m}}{m} - \sum_{j\leq k} w_j H(\mu_j)
\end{equation}
to $x$, where $w_j$ is the proportion of indices $i\leq N$ sampled from $\mu_j$ and $H(\mu_j)$ are the 
entropies of the sources. As detailed in the proofs, the term 
$-\sum_{i\leq k} w_i H(\mu_i)$ is the normalized log-likelihood that the model 
$I$ itself assigns to a typical sample $x$, and it represents the 
true likelihood of the data. Therefore $L(H_0,x)$ is a sum of a true likelihood, and an error term 
which decreases with increasing $m$. The log-likelihood 
(\ref{eq:intro_right_likelihood}) is achieved on an HMM that has 
emission distributions $\mu_i$ identical to those of $I$, and the 
probability of a state change in this HMM is of order $\frac{1}{m}$. 

In view of this, the main difficulty resolved in this paper, and the main technical contribution, consists in showing that if a fixed HMM 
$H$ has emission distributions that significantly differ from $\mu_1, \ldots, \mu_k$, then the 
log-likelihood it assigns to $x$ is lower then (\ref{eq:intro_right_likelihood}). We remark that due to 
the the hidden states, the likelihood function of an HMM is a complicated quantity which 
is usually controlled implicitly, see the discussion in 
\cite{douc2011}. On the other hand, in this paper we show that by appropriate use of  
type theories (for both the model $I$ and for a Markov chain) we can give explicit bounds on the likelihood 
for finite $N$. While type theory is a well known information theoretic tool, the particular combination 
of arguments that allows us to control the likelihood of an HMM is new.

To use type theory we will introduce the second moments of the model $I$ and the HMM. Roughly speaking, 
for each $a,b \in \groundX$ and a random vector $X$, the second moment $M_X(a,b)$ is the probability 
that $x_i = a$ and 
$x_{i+1} = b$ averaged over all $i$. The second moment captures a basic temporal structure of the 
process. The main technical result of the paper, Theorem \ref{thm:main_thm}, shows that if the second 
moment of an HMM $H$, denoted $M_H$, differs from the moment for the model $I$, $M_I$, then for most 
samples $x$ from $I$, the likelihood $L(H,x)$ will be low. Combined with additional arguments, 
this will imply that the maximum likelihood HMM will have the correct second moment. 

It is now natural to ask how much information does the second moment $H$ contain about the emission 
distributions $\nu_i$ of $H$? In particular, is it true that if $M_I = M_H$ then 
the model $I$ and $H$ have the same set of emission distributions? 
In general, the answer to this question is negative. Elegant 
counterexamples can be found in \cite{chung1996} (see also 
\cite{AHK12}). However, it is also well known and easy to see that the column space of the second 
moment matrices is spanned by the emission distributions. We will see that a similar statement holds 
for our definitions of moments, which somewhat differ from the classical ones. Therefore, if $M_I$ and 
$M_H$ are known, we can reconstruct the $k$-dimensional subspaces $span\Set{\mu_j} \subset \RR^{|
\groundX|}$ and $span\Set{\nu_j} \subset \RR^{|\groundX|}$ spanned by emissions of $I$ and $H$ 
respectively. Note that in order to specify a measure on $|\groundX|$ points one needs $|\groundX|-1$ 
parameters, but if one knows that the measure belongs to a given $k$-dimensional subspace, then only 
$k-1$ parameters are required. Since $k$ is typically much smaller then $|\groundX|$, this means that 
the second moment contains most of the information about the emissions (consider the case $k=2$ and 
$\Abs{\groundX} = 100$ for the sake of illustration). 

Finally, we note that our approach can extended to moments higher then two. Indeed, the main 
combinatorial tool used in this paper is type theory for second moments of Markov chains as 
developed in \cite{ccc87}, where higher moments analog is also presented. However, all the ideas 
necessary for such an extension are present already in the second moment case and in this paper we 
restrict our attention only to the second moments. 

The rest of this paper is organized as follows: In Section \ref{sec:literaure} we review the literature.
Section \ref{sec:defins} contains the definitions and the statements of the results, as well as a sketch of 
our main technical argument. We conclude by a discussion in Section \ref{sec:discuss}. For clarity of 
presentation, the full proofs are deferred to Section \ref{sec:sup_mat_proofs}.

\section{Related Work}
\label{sec:literaure}
As noted in the Introduction, real data often does not behave as a sequence generated by 
an HMM. Some aspects of this problem may be addressed is via the notion of Hidden \textit{semi-}Markov 
Models (HSMMs, see the survey \cite{hsmm_surv}). HSMM is an extension of an HMM which was developed in recent 
years to overcome a particular issue of state duration. In a Markov process, and hence in an HMM, the 
time the system stays in a given state is always a geometric, memoryless random variable, with an 
expectation that may depend on the state.  In a semi-Markov model, the duration of a stay in a given 
state is allowed to be an arbitrary random variable depending on the current state. 
While HSMMs were shown to be more suitable than the HMMs in a large 
variety of cases, this comes at a cost. Since one can not realistically model arbitrary duration times, 
one can either resort to parametric families of distributions that might be better suited to a 
particular application than the geometric variable, or one may consider arbitrarily distributed 
but bounded duration times. The first option requires expert knowledge of the application domain, while 
the second introduces a huge space of parameters and is still limited in what it can model (due to 
boundedness).  See \cite{hsmm_surv} for a detailed account of the advantages and the issues with HSMMs. 

The approach of this paper provides a different perspective on the issue of duration times. 
Indeed, while an HSMM provides a more general model of transitions between the states of the system than HMM, we 
show that if we want to estimate the source distributions, then 
under Interval Model assumptions we \textit{do not need} to model the transitions between the sates at all, 
and the simple HMM estimator suffices.  This has run time and sample complexity advantages, but more 
importantly  -- we are guaranteed an approximation to the true sources without the need to guess and to model
the transitions between the states.  

It is worth emphasizing that in some situations modelling the transitions is important. For instance, in 
speech recognition certain phonemes are much more likely to occur after certain other phonemes, and this 
transition information is important for the applications. However, in other situations, such as the financial 
time series and human activity series described in the Introduction, 
it is unlikely that there exists any stationary probabilistic model of the transitions. Hence 
it is important to know that an estimation procedure works for any, possibly non-stationary or 
non-probabilistic transition mechanism, as expressed by the Interval Model and guaranteed by our 
results. 

A problem setup somewhat similar to the Interval Model was recently investigated in \cite{khaleghi}, 
in the context of change point detection methods. Similarly to the Interval Model, in the model of 
\cite{khaleghi} the data is composed from intervals, and each interval is generated by one of $k$
sources. Moreover, the sources there can be arbitrary stationary processes, which is significantly more 
general than the independent processes which we consider in this paper. 
However, the results of \cite{khaleghi} hold only in the asymptotic regime where the number of intervals 
is fixed and the number of samples $N$ goes to infinity. This means that the length of each individual 
interval is required to go to infinity with $N$. This makes the problem simpler, since in this regime one can 
essentially learn the source from a single interval. In contrast, in the Interval Model we require the 
intervals to be of minimal length $m$, but we do not require the lengths to go to infinity with $N$, and 
our approximation results hold for any fixed $m \geq 2$. This regime requires the estimator to combine the 
information from \textit{all} the intervals in the sample to estimate the sources, and our approach uses
methods completely different from those of \cite{khaleghi}.

We now turn to a discussion of the literature related to the more technical aspects of this paper.
The classical consistency result for HMMs, \cite{baum1966}, \cite{petrie}, states that  if $x = 
(x_i)_{i=1}^{\infty}$ is an infinite sequence generated by an HMM $H$, and $H_n$ is a sequence of 
maximum likelihood estimators for the growing sequences $(x_i)_{i=1}^{n}$, then $H_n$ converges to $H$ 
with probability 1 (over $x$). A key technical component of these results is an extension of the
Shannon-McMillan-Breiman Theorem. This extension deals with the asymptotic behaviour of the 
likelihood assigned by a given HMM to a sequence generated by a different HMM. The original results were 
formulated and proved for finite state HMMs with a finite value alphabet $\groundX$, and with additional light 
restrictions on $H$. More recently, several results have appeared which extend the consistency theorem to HMMs 
with more general state and value spaces, and investigate the conditions under which such extensions are 
possible. See for instance \cite{legland_mevel}, \cite{douc2011}. 

Our result can also be viewed as an extension of the consistency theorem, but in a different 
direction. We consider only finite state HMMs and finite value spaces $\groundX$, but we 
do not assume that $x$ is generated by an HMM. The study of such questions, known as \textit{misspecification}
results, started only recently. The results in \cite{mevel} characterize the behaviour 
of a maximum likelihood estimator for HMMs when $x$ is generated by general ergodic processes satisfying 
some mixing-type conditions. In particular it is shown that the sequence $H_n$ of maximum likelihood 
estimators converges to an HMM $H$ such that the limiting Kullback-Leibler divergence between $H$ and 
the
process generating $x$ is minimal (see \cite{mevel}, Section IV). However, neither the result nor its 
proof supply any information about what the minimizing $H$ actually is, and are therefore of a limited 
practical value. We note, however, that such a limitation is in fact unavoidable, due to the generality 
of the setup. If all we know about $x$ is that it is generated by a general ergodic process, it is 
unlikely that anything concrete can be said about $H$. On the other hand, in this paper we assume a 
specific structure of the process $x$, namely that it is generated by an interval model, and we show 
that in this case, the source distributions of $H$ approximate those of $I$. Therefore, our result can 
also be viewed as a statement about the properties of the minimizer $H$ for the case when $x$ is 
generated by $I$. 

In all of the above mentioned work on consistency and misspecification, the assumption of ergodicity of the 
process generating $x$ plays a crucial role and the underlying proof methods rely heavily on this assumption. 
It is therefore interesting to note that in this work we do not require the Interval Model to be ergodic. The 
details on the relation between the Interval Model and ergodicity are given in Section \ref{sec:IM_ergodicity}. Here 
we mention that in contrast to the existing methods, our approach \textit{provides inequalities that are valid for 
finite $N$} rather than asymptotic results, which allows us to avoid the global ergodicity assumption and to 
work in the more general adversarial setting. 

 Moments of the data play an important role in our approach.  In recent years, moments of the data
have been used for parameter estimation in various mixture models. For instance, in \cite{AroraBSVD},
\cite{AroraGHMMSWZ13}, it was shown that for several types of mixture models, the underlying 
distributions 
$\mu_j$ can be inferred from the second moment of the data under an ``anchor words'' assumption on  
$\mu_j$s. In \cite{AHK12} it was shown that for a sufficiently large number of samples and under lighter 
assumptions on $\mu_j$, the third moment of the data can be used to reconstruct $\mu_j$ for a variety of 
mixtures, including the HMM. Note that the use of moments in this paper is different. Our estimator is 
the classical maximum likelihood estimator rather than an estimator based on moments. We use moments 
only as a tool to show that properties of the estimator approximate the properties of the true model. 

Finally, we make essential use of type theory for Markov chains. The results we use were obtained in 
\cite{ccc87}, where second order and higher order type theory is developed.

\section{Definitions and Results}
\label{sec:defins}
In Sections \ref{sec:def_res_prelim},  \ref{sec:models} and \ref{sec:moments} we introduce the notions 
necessary to state the results. Section \ref{sec:results} contains the statements and outlines of the proofs. 

\subsection{Preliminaries}
\label{sec:def_res_prelim}
For a finite set $S$, denote by $\Delta_{S}$ the set of all probability measures 
on $S$. For any two probability distributions $\mu,\nu \in \Delta_{\groundX}$, define the entropy and 
 the Kullback-Leibler 
divergence by
\begin{equation}
H(\mu) = -\sum_{a \in \groundX} \mu(a) \log \mu(a) \\
\end{equation}
and 
\begin{equation}
D(\nu|\mu) = \sum_{a \in \groundX} \nu(a) \log \frac{\nu(a)}{\mu(a)}.
\end{equation}
The total variation distance between $\mu,\nu \in \Delta_{\groundX}$ is given by 
\begin{equation}
\norm{\mu - \nu}_{TV} = \sum_{s \in S} \Abs{\mu(s)- \nu(s)}.
\end{equation}

\subsection{Models}
\label{sec:models}
An Interval Model is a tuple 
$I = I\left(\Set{I_l}_{l \in \NN},\Set{\mu_i}_{i=1}^k,\tau, m \right)$, where
$I_l$ is a sequence of consecutive intervals, $I_l = [b_l,e_l] 
\subset \NN$, such that $b_1 = 1$, and $b_{l+1} = e_l + 1$ for all 
$l$,  $\mu_i$ are probability measures on a fixed finite ground set 
$\groundX$, $\tau: \NN \rightarrow \Set{1,\ldots,k}$ is an assignment 
of distributions to intervals, and $m>0$ is such that $|I_l|\geq m$ 
for all $l \in \NN$. We say that a sequence of random variables with 
values in $\groundX$, $X = X_1,X_2, \ldots$ ,
is distributed according to interval model $I$, denoted $X \sim I$, if $X_i$ are 
independent and for every $l \in \NN$ and $i \in I_l$, $X_i$ has 
distribution $\mu_{\tau(l)}$.

For any finite $N$, the let the weights $\Set{w_j}$ be the proportions of each of the states 
$\mu_j$ in the data. Specifically, define 
\begin{equation}
\label{eq:k_j_def}
K_j(N) = \Set{i \leq N \setsep i \in I_l \mbox{\spaceo and \spaceo} \tau(l) = j } 
\end{equation}
to be the set of indices $i \leq N$ such that $X_i \sim \mu_j$ and set 
\begin{equation}
\label{eq:weights_def}
w_j = w_j(N) = \frac{1}{N}\Abs{K_j(N)}.
\end{equation}
Note that $w_i$ depends on $N$. For brevity of the notation this dependence is always assumed but not 
explicitly written. 

Throughout the paper we assume for convenience that $m>2$.

For each time $i \in \NN$ we define $\kappa(i)$ to be the index of the 
distribution of $X_i$, meaning $\kappa(i) = \tau(l)$ where $l$ is such that $i \in I_l$. 

A Hidden Markov Model, HMM, is a tuple $H = H\left(S,\{\nu_i\}_{1}^k, 
\{p_{ij}\}_{i,j=1}^k\right)$ where $S = \Set{1,\ldots,k}$ is a state 
space, $\nu_i$ are corresponding emission probabilities, and 
$p_{ij} = \Prob{S_{t+1} = j | S_t = i}$ for the Markov chain $S_1,S_2,S_3, \ldots$ of the states. 

For a sequence $x = (x_1,\ldots,x_{N+1})$, the log-likelihood of $x$ 
under the HMM $H$ with initial distribution $\pi$ is defined by 

\begin{align}
\label{eq:HMM_x_likelihood}
&L(x,H,\pi) = \\
& =\frac{1}{N+1} \log  \Brack{\sum_{s = s_1,\ldots,s_{N+1}} 
  \pi(s_1) \cdot \prod_{i=1}^N p_{s_{i},s_{i+1}} 
  \prod_{i=1}^{N+1} \nu_{s_i}(x_i)},  \nonumber
\end{align}
where the sum is over all possible paths of length $N+1$ of the underlying Markov 
chain.

\subsection{Moments}
\label{sec:moments}
For a sequence
 $x = (x_1,\ldots,x_{N+1})$, 
the second moment is a probability distribution $M(x) \in \grmeas$, defined by 
\begin{equation}
\label{eq:data_moment_def}
M(x)(a,b) = \frac{1}{N}\Abs{\Set{i\leq N \setsep x_i = a \wedge 
x_{i+1} =b}}
\end{equation}
for all $a,b \in \groundX$.
The second moment describes the frequencies of observing each pair of symbols $a,b$ consecutively.
For a random vector $X=(X_1,\ldots,X_{N+1})$, the second 
moment is the expectation of moments over all realizations of $X$,
\begin{equation}
M_X = \Exp_{x \sim X} M(x).
\end{equation}
For instance, if $X_i$ are independent and have the same distribution $\mu$, then
$M_X(a,b) = \mu(a) \cdot \mu(b)$.  

To obtain an expression for the second moments for interval model $I$, define for fixed $N$
\begin{equation}
\label{eq:I_transition_counts}
c_{rl} = \Abs{\Set{ i < N+1  \setsep \kappa(i) = r \wedge \kappa(i+1) = l  }}.
\end{equation}
$c_{rl}$ counts the transitions from state $r$ to state $l$ in the model, up to time $N+1$.  
Then, if $X =(X_1, \ldots, X_{N+1}) \sim I$, we have
\begin{equation}
\label{eq:I_moment_def}
M_X(a,b) = \frac{1}{N} \sum_{r,l \leq k} c_{rl} \mu_r(a) \cdot \mu_l(b).
\end{equation}

Next, to state our technical result, Theorem \ref{thm:main_thm}, we will require a definition of a \textit{generalized second moment} of an HMM. To motivate this definition, let us first write (\ref{eq:I_moment_def}) in a slightly 
different form. Denote for every $r,l \leq k$, $u_{rl} = \frac{c_{rl}}{N}$ and set 
$u_{r} = \sum_{l\leq k} u_{rl}$.  Then one can write (\ref{eq:I_moment_def})  as
\begin{align}
\label{eq:moment_expanded}
&M_{X}(a,b) =   \\
&\left(\begin{array}{lll}
u_1 \mu_1(a), \ldots, u_k \mu_k(a) 
\end{array}\right)
\left(\begin{array}{lll}
\hdots & \hdots & \hdots \\
\hdots & u_{ij} & \hdots \\
\hdots & \hdots & \hdots
\end{array}\right)
\left(\begin{array}{l}
\mu_1(b) \\
\vdots \\
\mu_k(b) 
\end{array}\right).  \nonumber
\end{align}
Equivalently, we have 
\begin{equation}
\label{eq:generalized_moment_motive}
M_{X}(a,b) = \phi_a \cdot U \cdot \chi_b,
\end{equation}
where $\phi_a = (u_1 \mu_1(a), \ldots, u_k \mu_k(a))^{T} \in \RR^k$, $\chi_b = (\mu_1(b), \ldots, 
\mu_k(b)) \in \RR^k $ and $U$ is the $k \times k$ matrix $U = (u_{ij})$. 

Now, given an HMM $H = H(S,\{\nu_i\}_{1}^k, 
\{p_{ij}\}_{i,j=1}^k)$, and a set of \textit{arbitrary} vectors $\phi 
= \{\phi_a\}_{a \in \groundX} \in \RR^k$, we define the 
\textit{generalized second moment} of $H$ as a matrix 
$M_{\phi,H} \in \RR^{\Abs{\groundX} \times \Abs{\groundX}}$ given by
\begin{equation}
\label{eq:generalized_moment_def}
M_{\phi,H}(a,b) = \phi_a \cdot p \cdot \chi_b,
\end{equation}
where analogously to (\ref{eq:generalized_moment_motive}) we have $\chi_b = (\mu_1(b), \ldots, 
\mu_k(b)) \in \RR^k $, but $\phi_a$ are arbitrary. The reasons for requiring this definition will become apparent during the proof of Theorem \ref{thm:main_thm}.

We call a set of vectors $\phi 
= \{\phi_a\}_{a \in \groundX}$ as above \textit{proper} if all the 
entries of all $\phi_a$ are non-negative, and 
\begin{equation}
\sum_{a \in \groundX} \sum_{j\leq k} \phi_a(j)  = 1. 
\end{equation}
If $\phi$ is a proper system, define 
a probability measure $d_{\phi}$ on $\groundX$ by 
\begin{equation}
d_{\phi}(a) = \sum_{j\leq k} \phi_a(j). 
\end{equation}

We conclude this section by stating the connection between column spaces of $M_X$ , $M_{\phi,H}$, 
and spaces spanned by $\{\mu_j\}_{j\leq k}$ and $\{\nu_j\}_{j\leq k}$ respectively. Note that for any 
matrix $M$, the column space of $M$ coincides with the image of $M$, $Im(M)$ as an operator 
$\RR^{\groundX} \rightarrow \RR^{\groundX}$.

\begin{lem} 
\label{lem:measures_span_moments}
\begin{enumerate}
\item
If $X \sim I$ for an interval model $I$. Then $Im(M_X) \subset span\{\mu_j\}_{j\leq k}$.
\item
For an HMM $H$ and an arbitrary set $\{\phi_a\}_{a \in \groundX}$, $Im(M_{\phi,H}) \subset span\{\nu_j\}_{j\leq k}$.
\end{enumerate}
\end{lem}
The proof is given in Section \ref{sec:proof_lem_meas_span_moments}.

Finally, for any $M \in \grmeas$, define the left and right marginalizations, 
$\bar{M},\dbar{M} \in \grmeassingle$ by 
\begin{equation}
\bar{M}(a) = \sum_{b \in \groundX} M(a,b) , \spaceo \dbar{M}(a) = \sum_{b \in \groundX} M(b,a).
\end{equation}

\subsection{Results}
\label{sec:results}
As discussed in the Introduction, the first part of argument consists in showing that there is an HMM 
$H$ that assigns high likelihood to most of the samples from $I$. This is formalized in the following 
Lemma.

\newcommand{\Nmin}{N_{min}}

Given an Interval Model $I = I\left(\Set{I_l}_{l \in \NN},\Set{\mu_i}_{i=1}^k,\tau, m \right)$,
for any $N>0$ we define 
\begin{equation}
N_{min} = \min_{j \leq k} w_j \cdot N,
\end{equation}
with $w_j$ as defined in (\ref{eq:weights_def}). 

\begin{lem}
\label{lem:right_path_full}
For any set of probability distributions $\Set{\mu_j}_{j=1}^k$, there exists a function $\eps : \NN 
\rightarrow \RR$ such that $\lim_{N \rightarrow \infty} \eps(N) \rightarrow 0$
and such that the following holds: \newline
For any Interval Model $I = I\left(\Set{I_l}_{l \in \NN},\Set{\mu_i}_{i=1}^k,\tau, m \right)$, 
there is an HMM $H$ and an initial 
distribution $\pi$, such that for every $N >0$, if $X=(X_1,\ldots,X_N) \sim I$ then   
with probability at least $1 - \eps(\Nmin)$, 
\begin{equation}
\label{eq:lem_high_likelihood}
L(X,H,\pi) \geq - \frac{\log 2km}{m} - \sum_{j} w_j H(\mu_j) - \eps(\Nmin).
\end{equation}
\end{lem}
The proof is given in Section \ref{sec:proof_lem_right_path_full}. We take a moment to discuss 
the particular dependence on $N$ exhibited in the above Lemma. The fact that the error term 
$\eps(\Nmin)$ in (\ref{eq:lem_high_likelihood}) depends on $\Nmin$ rather than $N$ means that 
in order for $\eps(\Nmin)$ to be small, the interval $[1,\ldots,N]$ needs to contain a sufficient number of 
samples from every one of the distributions $\mu_1,\ldots, \mu_k$. As will be evident from the proof, this 
assumption is necessary to obtain (\ref{eq:lem_high_likelihood}). On the other hand, the function $\eps$
is completely determined by the distributions $\Set{\mu_j}_{j=1}^k$. In other words, in order to control 
the error in (\ref{eq:lem_high_likelihood}) for a model $I$, we only need to know its distributions, and in 
particular $\eps$ does not depend on the particular interval structure of the model. 

We are now ready to state the main result of this paper. 
Let $X=(X_1,\ldots,X_N)$ be generated by an Interval Model $I$. For an HMM $H$ define 
\begin{equation}
\label{eq:main_condition_on_phi_a}
D  = D(H) = \inf_{\phi \in P} \norm{M_X - M_{\phi,H}}_{TV} - \frac{3}{m}, 
\end{equation}
where 
\begin{equation}
\label{eq:P_defin}
P = \Set{\phi \setsep \mbox{$\phi$ is proper and } \norm{d_{\phi} - \bar{M}_X}_{TV} \leq  
 \frac{3}{m}}.
\end{equation}
In other words, $D$ measures how well $M_X$ can be approximated by a generalized moment 
$M_{\phi,H}$ where $\phi$ can be any proper system with $d_{\phi}$ close to the marginal $\bar{M}_X$. 
To gain some intuition into this quantity, consider the case where $D$ is small, and  
$M_X$ has the maximal rank, $k$. Then, standard matrix perturbation theory results 
imply that $Im(M_X)$ is close to $Im(M_{\phi,H})$ and hence $span\{\mu_j\}_{j\leq k}$ is close to 
$span\{\nu_j\}_{j\leq k}$ by Lemma \ref{lem:measures_span_moments}.   Note also that the set $P$ in 
(\ref{eq:P_defin}) is 
non-empty. Indeed, the proper system $\phi$ defined in (\ref{eq:generalized_moment_motive}) satisfies 
$d_{\phi} = \bar{M}_X$.

We assume throughout the paper that $D \geq 0$, which amounts to considering only the cases 
where $M_{\phi,H}$ at least somewhat differs from $M_X$.

\begin{thm} 
\label{thm:main_thm}
There is a constant $c>0$ such that for any set of probability distributions $\Set{\mu_j}_{j=1}^k$, there 
exists functions $\eps$ and $r$ such that $\lim_{N \rightarrow \infty} \eps(N) \rightarrow 0$, 
$\lim_{N \rightarrow \infty} r(N) \rightarrow c$ and the following holds: \newline
For any Interval Model $I = I(\Set{I_l}_{l \in \NN},\Set{\mu_i}_{i=1}^k,\tau, m )$, 
and HMM $H = H(S,\{\nu_i\}_{1}^k, \{p_{ij}\}_{i,j=1}^k)$, if $X=(X_1,\ldots,X_N)$ is a sample from $I$ then
with probability at least 
$1 - \eps_{\Nmin}$ 
over $X$, 
for every initial distribution $\pi$, 
\begin{equation}
\label{eq:likelihood_in_main_thm}
L(x,H,\pi) \leq -D^2 - \sum_{j} w_j H(\mu_j) + \eps(\Nmin),
\end{equation}
where $D$ is as defined in (\ref{eq:main_condition_on_phi_a}).
\end{thm}
The proof is given in Section \ref{sec:proofs}. Here we briefly describe the main idea of the proof. Fix an Interval Model $I$, an HMM $H$, and $N>0$. Define a neighbourhood 
$U \subset \grmeas$ of $M_X$ by 
\begin{equation}
U = \Set{ M \in \grmeas  \setsep \norm{M - M_X}_{TV} \leq \frac{3}{m}},
\end{equation}
and denote by $O$ the set of all sequences $x = (x_1, \ldots, x_N)$  such that $M(x) \in U$. 
Roughly speaking, the proof of (\ref{eq:likelihood_in_main_thm}) can be seen as a combination of two different 
uses of type theory. First, the type theory of Markov chains can be used to show that if an HMM $H$ satisfies 
(\ref{eq:main_condition_on_phi_a}), then the likelihood  assigned to the set $O$ by $H$ is at most
 $2^{-N D^2}$,
\begin{equation} 
\Probu{H}{O} = \sum_{x \in O} 2^{N L(x,H,\pi)} \leq 2^{-N D^2}.
\end{equation}
On the other hand, type theory for independent sequences together with additional concentration 
results can be used to show that $O$ contains a subset $X^l \subset O$ of size at least 
$2^{N \cdot\Brack{\sum_{j} w_j H(\mu_j)}}$ such that all $x \in X^l$ are equiprobable with respect to $X$ and 
$X^l$ is of nearly full measure, $\Probu{X}{X^l} \geq 1 - \eps$. 
Combing these two statements, one obtains
\begin{equation}
\label{eq:proof_sketch1}
\frac{1}{|X^l|} \sum_{x \in X^l} \Probu{H}{x} \leq \frac{1}{|X^l|} \Probu{H}{O} \leq 
2^{-N \Brack{ D^2  + \sum_{j} w_j H(\mu_j)} }.
\end{equation}
Note that (\ref{eq:proof_sketch1}) is in fact an averaged version of (\ref{eq:likelihood_in_main_thm}). 
The corresponding high probability formulation can be easily obtained via Markov's inequality. 

Finally, we state our main result about the behaviour of the maximum likelihood HMM estimator on samples
of model $I$. For $\delta >0$, let $\mathcal{H}_{\delta}$ be the set of HMMs for which transition and 
emission probabilities are bounded below by $\delta$,
\begin{equation}
\mathcal{H}_{\delta} = \Set{ H(S,\{\nu_i\}_{1}^k, \{p_{ij}\}_{i,j=1}^k)}
\end{equation}
where $\nu_i$ and $p_{ij}$ satisfy 
\begin{equation}
\nu_i(x) \geq \delta \spaceo \forall i\leq k, x \in \groundX \mbox{ and } p_{ij} \geq \delta \spaceo \forall i,j\leq k. 
\end{equation}

In what follows we assume that the HMM guaranteed by Lemma \ref{lem:right_path_full} is in $\mathcal{H}_{\delta}$.  This is equivalent to the following:
\begin{equation}
\label{eq:delta_bound}
\delta \leq \frac{1}{m} \mbox{ and } \delta \leq \mu_j(x) \spaceo \forall j\leq k, x \in \groundX.
\end{equation}

\begin{thm} 
\label{thm:full_statement}
Fix the distributions $\Set{\mu_j}_{j=1}^k$ and $m>0$. For any $\delta$ satisfying (\ref{eq:delta_bound})
 there is a constant $c>0$, a function $r$ such that $\lim_{N \rightarrow \infty} r(N) \rightarrow c$ and 
a sequence $\eps_N$ with $\lim_{N \rightarrow \infty} \eps_N \rightarrow 0$ such that 
the following holds: \newline
Let $X=(X_1,\ldots,X_N)$ be a sample from Interval Model
 $I = I(\Set{I_l}_{l \in \NN},\Set{\mu_i}_{i=1}^k,\tau, m )$. 
Let $H$ be a maximum likelihood estimator in $\mathcal{H}_{\delta}$ for the sequence $X$. Then with probability at least $1 - \eps_{\Nmin}$ over $X$, 
\begin{equation}
\label{eq:full_statement_d_h}
 D(H) \leq \sqrt{\frac{\log 3km}{m}}.
\end{equation}
\end{thm}


\begin{cor}
\label{cor:limsup}
Let $X=(X_N)_{N=1}^{\infty}$ be an infinite sample from an infinite Interval Model
 $I = I(\Set{I_l}_{l \in \NN},\Set{\mu_i}_{i=1}^k,\tau, m )$. 
Let $H_N$ be a sequence of maximum likelihood estimators for the sequences $(X_1,\ldots,X_N)$. 
Then with probability $1$,
\begin{equation}
\label{eq:full_statement_d_h_infty}
 \limsup_{N \rightarrow \infty} D(H_N) \leq \sqrt{\frac{\log 3km}{m}}.
\end{equation}
\end{cor}

We now make a few remarks about the proof. The proof of Theorem \ref{thm:full_statement} is obtained by an 
application of Theorem \ref{thm:main_thm} to an appropriate (multiplicative) $\eps$-net inside the set 
$\mathcal{H}_{\delta}$ and by the union bound.  For this approach to work the log-likelihood 
needs to be a Lipschitz function of the HMM $H$. This is guaranteed by the assumption $H \in 
\mathcal{H}_{\delta}$, with the Lipschitz constant depending on $\delta$. This assumption is 
common in the literature and is used for similar purposes, although it is usually used somewhat differently. 
Moreover, this assumption can be easily removed if we let $N \rightarrow \infty$ in the Theorem statement,
as formalized in Corollary \ref{cor:limsup}. Note that in contrast to existing consistency results, 
since we do not assume ergodicity of $I$, the sequence $H_N$ in the statement of Corollary \ref{cor:limsup}
does not necessarily converge to a limit. 
Nevertheless, inequality (\ref{eq:full_statement_d_h_infty}) holds.
Another remark concerns the magnitude of $N$ required for (\ref{eq:full_statement_d_h}) to hold with high
probability. While the size of $\eps$-net in $\mathcal{H}_{\delta}$ is exponential in 
$k$ and $\Abs{\groundX}$, the probability of error in Theorem \ref{thm:main_thm} is essentially exponentially small 
in $N$. Therefore for the union bound to hold, it suffices for $N$ to be polynomial in $k$ and $\Abs{\groundX}$.
The complete proof is given in Section \ref{sec:proof_of_full_statement}.

\section{Discussion}
\label{sec:discuss}
In this work we considered time series generated by a finite state system, with an assumption that 
the states are somewhat \textit{persistent}, in the sense that the system stays at a state at least 
$m$ time units before the state changes. The advantage of such an 
assumption is that for a variety of systems it may be fairly realistic, and that it is minimal in the 
sense 
that we do not attempt to model the transition mechanism of the system between different states. Indeed, 
we show that for \textit{any} such mechanism, the distributions of the sources can still be 
approximated. An apparent paradox of this result is that we show that the approximation can be done 
using an HMM estimator, and HMM estimator \textit{does} assume a particular transition mechanism between 
the states. The resolution of this paradox, and the reason that the approach works, is that when the 
states have the persistence property, the transition mechanism provably has little influence on certain 
statistics (the likelihood and the moments, for this paper) of the system. 

From a purely technical perspective, one possible extension of this work would be to relax the 
assumption that the system should spend at least $m$ time units at each state. It should be enough for a 
system to 
satisfy this assumption most of the time rather than strictly every time it enters a new state. We 
believe such an extension can be proved using the approach developed in this paper. Another possible 
extension is to replace the discrete emission distributions used in this paper by some continuous 
class, such as the Gaussians. 

From a more general perspective, we believe that the idea of replacing fully generative models by more 
adversarial settings may be extended to other estimation problems. Consider for instance 
topic modelling, where topics are distributions and documents are samples from mixtures of topics. 
By far the most popular generative model used to learn topics is the Latent Dirchlet Allocation, (LDA,
\cite{LDAorig}). LDA proposes a particular mechanism by which the documents are generated from the 
topics. LDA often performs extremely well in practice, despite the fact that real documents 
are clearly not generated by the LDA mechanism. This indicates that, analogously to our HMM results, there 
might exist 
some persistence properties of the topics which would guarantee that the topics are recovered by maximum 
likelihood LDA despite the fact that the documents were not generated by LDA. Identification of such 
persistence properties could contribute, for instance, to model independent (or less model-dependent) 
definitions of topics.

\section{Proofs}
\label{sec:sup_mat_proofs}

\subsection{Interval Model and Ergodicity}
\label{sec:IM_ergodicity}
In this section we discuss the relation between the Interval Model and ergodicity. For the purposes of this discussion, a process $X = (X_1, X_2,\ldots, X_N, \ldots)$ is ergodic if there exists 
a distribution $\mu$ on $\groundX$ such that for every $f : \groundX \rightarrow \RR$, 
\begin{equation}
 \lim_{N \rightarrow \infty} \frac{1}{N} \sum_{i=1}^N f(X_i) = \int f(x) d\mu(x)
\end{equation}
with probability $1$ over $X$. Ergodicity means that space integrals with respect to $\mu$ can be recovered 
by a time average over a single trajectory of the process. If the process $X$ is generated by an interval 
model, then for every $N$ we can write 
\begin{equation}
\label{eq:I_ergodicity}
 \frac{1}{N} \sum_{i=1}^N f(X_i)  = \sum_{j \leq k} w_j \frac{1}{\Abs{K_j}} \sum_{i \in K_j}^N f(X_i), 
\end{equation}
with the sets $K_j$ as defined in (\ref{eq:k_j_def}). Since $X_i$ with $i \in K_j$ are independent, if  $\Abs{K_j} \rightarrow \infty$, then we have 
\begin{equation}
\label{eq:I_ergodicity_single_comp}
\frac{1}{\Abs{K_j}} \sum_{i \in K_j}^N f(X_i) \rightarrow \int f(x) d\mu_j(x)
\end{equation}
by the law of large numbers. Therefore (\ref{eq:I_ergodicity}) converges if and only if for each $j\leq k$,
$w_j(N)$ converges to some limiting value $\hat{w}_j$ with $N \rightarrow \infty$, in which case the 
limiting measure is $\mu = \sum_{j \leq k} \hat{w_j} \mu_j$.  This situation demonstrates well the general 
line of reasoning used in this paper. We will assume that 
$\Abs{K_j}$ is large enough for integrals with respect to each component to converge (this corresponds to large $\Nmin$ in the statements of the results), as in 
(\ref{eq:I_ergodicity_single_comp}), but we will \textit{not} require the more global condition that the 
weights $w_j(N)$ converge.

\subsection{Preliminaries} 
\label{sec:proof_lem_meas_span_moments}

In what follows it will be convenient to use the tensor notation for operators -- for any two vectors $v,w \in 
\RR^{\groundX}$, $v \otimes w$ is a rank 1 linear operator $\RR^{\groundX} \rightarrow \RR^{\groundX}$, 
which acts by $(v \otimes w)(u) = \inner{u}{v} \cdot w$ for all $u \in \RR^{\groundX}$. In particular, for 
$a,b \in \groundX$ we have $\inner{(v \otimes w) \delta_a}{\delta_b} = v(a) \cdot w(b)$.  

For instance, if $X= (X_1, \ldots, X_{N+1})$ all $X_i$ are independent with distribution $\mu$ then 
\begin{equation}
 M_X = \mu \otimes \mu. 
\end{equation}
If $X = (X_1, \ldots, X_{N+1}) \sim I$ is a sample from the interval model, we can write 
\ref{eq:I_moment_def} equivalently as 
\begin{equation}
\label{eq:I_moment_def_tensor}
M_X = \frac{1}{N} \sum_{r,l \leq k} c_{rl} \mu_r \otimes \mu_l.
\end{equation}

An important property of the interval model is that the number of transitions between \textit{different} 
states in the model is small. We formalize it in the following Lemma. 
\begin{lem} For an interval model $I$ and $N>0$, for every $r,l \leq k$, let $c_{rl}$ be the state 
transition counts, as defined in (\ref{eq:I_transition_counts}). Then 
\begin{equation}
\label{eq:small_mixed_moments}
\frac{1}{N} \sum_{r \neq l} c_{rl} \leq \frac{1}{m}. 
\end{equation}
\end{lem}
\begin{proof}
Indeed, since interval length is at least $m$, the set $\Set{1,\ldots , N+1}$ contains at most $\ceil{ 
(N+1)/m}$ different intervals and hence at most $\ceil{ (N+1)/m} - 1$ transitions.
\end{proof}
It follows from (\ref{eq:small_mixed_moments}) that 
\begin{equation}
\label{eq:matrix_small_mixed_moments}
\norm{\frac{1}{N} \sum_{r \neq l } c_{rl} \mu_r \otimes \mu_l}_{TV} = 
\norm{M_X - \sum_{i \leq k} w_i \mu_i \otimes \mu_i}_{TV} \leq \frac{1}{m}.
\end{equation}
We refer to the expression 
\begin{equation}
\label{eq:pure_moment_def}
M_{X,pure} = \sum_{i \leq k} w_i \mu_i \otimes \mu_i
\end{equation}
as a \textit{pure} moment and 
to the expression 
\begin{equation}
\label{eq:mixed_moment_def}
M_{X,mixed} = \frac{1}{N} \sum_{r \neq l } c_{rl} \mu_r \otimes \mu_l
\end{equation}
as a \textit{mixed} moment. The pure
moment captures the contribution to the moment inside each interval, while the mixed moment captures the 
contribution from transitions between the intervals.

Finally, we prove Lemma \ref{lem:measures_span_moments}.
\begin{proof}[Proof of Lemma \ref{lem:measures_span_moments}]
The statement for $M_X$ follows directly from (\ref{eq:I_moment_def}). To show the statement for 
$M_{\phi,H}$, for any $\phi_a \in \RR^k$ and $i \leq k$ let $\phi_a(i)$ be the $i$-th coordinate of 
$\phi_a$.  For every $i \leq k$, define $\hat{\phi}_i \in \RR^{\groundX}$ by $\hat{\phi}_i (a) = \phi_a(i)$. 
Then by the definition, (\ref{eq:generalized_moment_def}), 
\begin{equation}
M_{H,\pi}(a,b) = \sum_{i \leq k} \sum_{j \leq k} p_{ij} \phi_a(i) \nu_j(b),
\end{equation}
and hence 
\begin{equation}
M_{H,\pi} = \sum_{i,j \leq k}  p_{ij} \hat{\phi}_i \otimes \nu_j.
\end{equation}
Since the image of $\hat{\phi}_i \otimes \nu_j$ is spanned by $\nu_j$, it follows that 
$Im(M_{H,\pi}) \subset span\{\nu_j\}_{j=1}^k$.
\end{proof}

\subsection{Proof of Lemma \ref{lem:right_path_full}}
\label{sec:proof_lem_right_path_full}
\begin{proof}[Proof of Lemma \ref{lem:right_path_full}]
Consider an HMM $H$ with $k$ states, $S = \Set{1,\ldots,k}$, with 
emission probabilities equal 
to those of the model $I$, $\mu_i$, and some transition matrix, $p_{ij}$.  In order to show the lower bound 
on $L$, it suffices to consider a single path of the HMM. Let $s=s_1,
\ldots,s_N$ be a sequence of states of $H$ that follows precisely the 
sequence of states in $I$, so that $s_i = \kappa(i)$. Let the initial 
distribution $\pi$ be a delta measure concentrated on the first state 
of $I$, $\kappa(1)$. Recall that the likelihood 
$L(x,H,\pi) = \frac{1}{N} \log \Probu{H,\pi}{x}$ is given by a sum 
(\ref{eq:HMM_x_likelihood}). The contribution of a single path $s$ in
this sum is 
\begin{equation}
\label{eq:true_path_likelihood}
\frac{1}{N} \sum_{i,j \leq k}  c_{ij} \log p_{ij} + 
\frac{1}{N} \sum_{j=1}^k  \sum_{i} \log \mu_j(x_i^j). 
\end{equation}
where $c_{ij}$ are the transition counts of the model $I$, as in 
(\ref{eq:I_transition_counts}), and $x_i^j$ are the entries of $X$, 
rearranged so that for all $i$, $x_i^j$ are entries sampled from $\mu_j$. 
Consider the second term first,
\begin{equation}
\frac{1}{N} \sum_{j=1}^k  \sum_{i} \log \mu_j(x_i^j) = 
 \sum_{j=1}^k  \frac{c_{jj}}{N} \frac{1}{c_{jj}} \sum_{i} \log \mu_j(x_i^j).  
\end{equation}
Clearly, by the law of large numbers,
$\frac{1}{c_{jj}} \sum_{i} \log \mu_j(x_i^j) \rightarrow -H(\mu_j)$
with $\Abs{c_{jj}} \rightarrow \infty$.  

Next, the first term in (\ref{eq:true_path_likelihood}),
\begin{equation}
\frac{1}{N} \sum_{i,k \leq k}  c_{ij} \log p_{ij}
\end{equation}
controls the underlying Markov chain probability of the path $s$. 
Since the total number of transitions between 
different states in the model is small, see (\ref{eq:small_mixed_moments}), this probability is large 
when $p_{ii}$ are close to $1$. In particular, by choosing 
$p_{ii} = 1 - \frac{1}{m}$ and $p_{ij} = \frac{1}{(k-1)m}$ for all 
$i$ and $j \neq i$, we obtain
\begin{eqnarray}
\frac{1}{N} \sum_{i,k \leq k}  c_{ij} \log p_{ij}  =  \\
\frac{1}{N} \sum_{i \leq k}  c_{ii} \log (1 - \frac{1}{m})  + 
\frac{1}{N} \sum_{i \neq j }  c_{ij} \log \frac{1}{(k-1)m}  \geq \\
-\frac{1}{m} - \frac{1}{m^2}   - \frac{1}{m} \log (k-1)m \label{eq:log_passage_lem_right_path} \geq \\ 
- \frac{\log 2km}{m}
\end{eqnarray}
where in line (\ref{eq:log_passage_lem_right_path}) we have used (\ref{eq:small_mixed_moments}) and the fact that 
\begin{equation}
\log ( 1- \frac{1}{m}) \geq -\frac{1}{m} - \frac{1}{m^2} 
\end{equation}
for $m\geq 2$.
\end{proof}

\subsection{Proof of Theorem \ref{thm:main_thm}}
\label{sec:proofs}

Let $x = (x_1, \ldots, x_N)$ be distributed according to an interval model $I$.
We first show that the empirical second 
moment of a sample $x$  is close to its expected second 
moment, $M_X$. 

\begin{lem}
\label{lem:empirical_mx_bound}
Let $X = (X_1, \ldots, X_N)$ be distributed according to the interval model 
$I = I(\Set{I_l}_{l \in \NN},\Set{\mu_i}_{i=1}^k,\tau, m )$. Denote 
\begin{equation}
\label{eq:tilde_w_def}
\tilde{w} = \min_{i \leq k} w_i.
\end{equation}
Then for every $\eps \geq 0$,
\begin{equation}
\label{eq:lem_empitrical_mx_bound}
\Probu{X}{ \norm{M(x) - M_X}_{TV}  \geq \eps + \frac{2}{m}} \leq 
2^{-c_1 \tilde{w}N \cdot \Brack{\frac{\eps^2}{|\groundX|^4}  -  \frac{ c_2 \log k}{\tilde{w}N} }},
\end{equation}
where $c_1,c_2>0$ are absolute constants. 
\end{lem}

Before we proceed with the proof, note the particular form of the error, $\eps + \frac{2}{m}$, in 
(\ref{eq:lem_empitrical_mx_bound}). As we show in what follows, the pure component of $M_X$, 
(\ref{eq:pure_moment_def}), can be approximated by the pure component of $M(x)$ to arbitrary precision, 
giving rise to the $\eps$ term in the error. However, the mixed component of $M_X$ will not necessarily 
be approximated well by the mixed part of $M(x)$, but has a small norm, 
(\ref{eq:matrix_small_mixed_moments}) and gives rise to the $\frac{1}{m}$ term in the error. 
To see why the mixed moment of $M_X$ might not be well approximated by the mixed part of $M(x)$, recall 
that we do not place any assumptions on the transitions between different states in the interval model. 
In particular, for $r \neq l$, $c_{rl}$ need not be large and consequently there may not be enough 
samples to recover $\mu_r \otimes \mu_l$.  

\begin{proof}[Proof Of Lemma \ref{lem:empirical_mx_bound}]
We first consider samples from each state $j$ separately and show that they approximate 
$\mu_j \otimes \mu_j$ well. This can be achieved by standard methods if the samples are independent. 
We therefore divide the indices into independent pairs. 
Let 
\begin{equation}
A_j = \Set{ i < N+1 \setsep \kappa(i) = j \mbox{ \spaceo and \spaceo} \kappa(i+1) = j}
\end{equation}
be the set of indices $i$ such that $X_i \sim \mu_j$ and set 
\begin{equation}
B = \Set{1, \ldots, N } \setminus \Brack{\cup_{j \leq k} A_j}
\end{equation}
to be the set of indices where the transitions between intervals occur. 
Divide $A_j$ into a set of odd pairs and even pairs as follows: 
\begin{eqnarray}
A_j^1 = \Set{ (i,i+1) \setsep i \in A_j \mbox{,\spaceo i is odd} }, \\
A_j^2 = \Set{ (i,i+1) \setsep i \in A_j \mbox{,\spaceo i is even} }. 
\end{eqnarray}
For instance, 
if $(1,2,3,4,5,6) \subset A_j$, then $(1,2),(3,4),(5,6)$ are odd 
pairs, and $(2,3),(4,5)$ are even. Then the  
pairs in each $A_j^t$ are mutually independent. Hence they can be 
considered i.i.d samples from the measure $\mu_j \times \mu_j$.  
To estimate how well independent empirical samples of $\mu_j 
\times \mu_j$ approximate $\mu_j \times \mu_j$, we use the 
Dvoretzky Kiefer Wolfowitz inequality, \cite{dkw}, which bounds the $\sup$ distance between the 
empirical and true distribution.  Specifically, if $\mu$ is a probability distribution on a set $S$ and 
$Y_1,\ldots,Y_N$ are independent samples from $\mu$, it follows from the  Dvoretzky Kiefer Wolfowitz inequality that
\begin{equation}
\Prob{ \sup_{s \in S}  \Abs{ \mu(s) - \Brack{\frac{1}{N} \sum_{i \leq N} \delta_{Y_i}}(s) } \geq \eps} \leq 
2 \exponent{ -2 N \eps^2}
\end{equation}
for every $\eps \geq 0$. We apply this with $S = \groundX \times \groundX$, $\mu = \mu_j \times \mu_j$ 
and $Y_i = (X_i,X_{i+1})$. For $a,b \in \groundX$, $j \leq k$ and $t \in \Set{1,2}$, let 
\begin{equation}
R_{j,t}(a,b) = \mu_j(a)\mu_j(b)  
-\frac{1}{|A_j^t|} \sum_{(i,i+1) \in A_j^t} \delta(X_i = a)\cdot\delta(X_{i+1} = b)
\end{equation}
be the difference between the empirical and the true measures. 
Then 
\begin{equation}
\label{eq:dkw_use}
\Prob{ \max_{a,b \in \groundX}\Abs{ R_{j,t}(a,b)  
}  \geq \eps/|\groundX|^2 } \leq 2^{-c |A_j^t| \frac{\eps^2}{|\groundX|^4}}.
\end{equation}
Since $\norm{R_{j,t}}_{TV} \leq \Abs{\groundX}^2 \sup_{a,b} \Abs{R_{{j,t}}(a,b)}$, we obtain
the total variation bound 
\begin{equation}
\label{eq:dkw_tv}
\Prob{ \norm{R_{j,t}}_{TV}
  \geq \eps } \leq 2^{-c |A_j^t| \frac{\eps^2}{|\groundX|^4}}.
\end{equation}
Since $|A_j^t| \approx \half w_j N$, the union bound over all $j,t$ implies that 
\begin{align}
\label{eq:dkw_union}
\Prob{ \max_{j,t} \norm{R_{j,t}}_{TV} 
  \geq \eps } &\leq 2k \cdot 2^{-c \tilde{w}N \cdot\frac{\eps^2}{|\groundX|^4}} \nonumber \\
  &\leq 
  2^{-c \tilde{w}N \cdot \Brack{\frac{\eps^2}{|\groundX|^4}  -  \frac{ c_1 \log k}{\tilde{w}N} }}  
\end{align}
for an appropriate absolute constant $c_1 > 0$.
Finally, note that 
\begin{align}
M(x)(a,b)  = & \frac{1}{N} \sum_{j \leq k} \sum_{t \in \Set{1,2}} \sum_{i \in A_j^t} 
  \delta_{x_i}(a) \delta_{x_{i+1}}(b)  \\
  & + \frac{1}{N} \sum_{i \in B} \delta_{x_i}(a) \delta_{x_{i+1}}(b)  \nonumber \\
   = &  \nonumber  \\
  & \sum_{j \leq k} \sum_{t \in \Set{1,2}}  \frac{|A_j^t|}{N} \frac{1}{|A_j^t|} \sum_{i \in A_j^t} 
    \delta_{x_i}(a) \delta_{x_{i+1}}(b) \nonumber  \\
  & + \frac{|B|}{N} \frac{1}{|B|} \sum_{i \in B} \delta_{x_i}(a) \delta_{x_{i+1}}(b) .\nonumber 
\end{align}
Therefore 
\begin{align}
M(x) - M_X = & \\
 & \sum_{j \leq k} \sum_{t \in \Set{1,2}}  \frac{|A_j^t|}{N} R_{j,t} \nonumber \\
 & - M_{X,mixed} + \frac{|B|}{N} \frac{1}{|B|} \sum_{i \in B} \delta_{x_i}(a). \nonumber
\end{align}
By (\ref{eq:small_mixed_moments}) and (\ref{eq:matrix_small_mixed_moments}), the total variation norm of 
the last two terms is bounded by $\frac{1}{m}$ each.
The first term is a convex combination, and therefore the claim of the lemma follows from 
(\ref{eq:dkw_union}).
\end{proof}

We proceed with the proof of Theorem \ref{thm:main_thm}. Define a set 
\begin{equation}
\label{eq:u_eps_def}
U = \Set{ 
M \in \Delta_{\groundX \times \groundX} \setsep 
   \norm{M - M_X}_{TV} \leq \frac{3}{m}
}.
\end{equation}
Lemma \ref{lem:empirical_mx_bound} (with $\eps = 1/m$) states that $M(x) \in U$ with high probability over $I$.  

Next, fix an HMM $H$ and let $M_H = M_{\phi,H}$ be 
the generalized second moment attaining the infimum in 
(\ref{eq:main_condition_on_phi_a}),
\begin{equation}
\label{eq:m_x_m_h_dist}
D = \norm{M_X - M_H}_{TV} = \inf_{\phi \in P} \norm{M_X - M_{\phi,H}}_{TV} - 3/m.
\end{equation}

Let 
\begin{equation}
O = \Set{x=(x_1, \ldots, x_{N+1}) \setsep M(x) \in U }
\end{equation}
be the set of all data sequences $x$ with $M(x) \in U$. 
Using the type theory for second moments of Markov chains, we will  
show that 
\begin{equation}
\label{eq:entropy_hmm_bound}
\Probu{H,\pi}{ O } \leq 2^{-N D^2},
\end{equation}
for every initial distribution $\pi$, where $D$ is given by (\ref{eq:m_x_m_h_dist}).
 Equivalently, 
under $H$, probability of observing a sequence $x$ with $M(x) \in 
U$ is at most $2^{-N D^2}$. We first prove
Theorem \ref{thm:main_thm} assuming (\ref{eq:entropy_hmm_bound}), and 
then prove (\ref{eq:entropy_hmm_bound}).

The inequality (\ref{eq:entropy_hmm_bound}) bounds the likelihood 
under $H$ of \textit{all} $x$ such that $M(x) \in U_{\eps}$.  To 
prove Theorem \ref{thm:main_thm}, we need to 
bound the likelihood $L(x,H,\pi)$ of \textit{individual} sequence $x$ 
produced by $I$. Let 
\begin{equation}
H(X) = \frac{1}{N} H(X_1, \ldots, X_N) =  \sum_{i \leq k} w_i H(\mu_i) 
\end{equation}
be the entropy of a sample $X = (X_1, \ldots, X_N)$ from $I$. 

The following lemma describes the type theory for for samples from model $I$. Recall that $\tilde{w}$ 
was defined in (\ref{eq:tilde_w_def}) and has the property that in a sample $X = (X_1,\ldots,X_N)$  from 
$I$ every source appears at least $\tilde{w} N$ times. 
\begin{lem} 
\label{lem:iid_aep}
Let $X = (X_1,\ldots,X_N) \sim I$. Then there exist a subset $G \subset \groundX^N$ of sequences such that 
\begin{enumerate}
\item
For every $x^l= (x_1,\ldots,x_N) \in G$,
\begin{equation}
\label{eq:xl_close_to_exp}
\Abs{-\frac{1}{N} \log P_X(x^l) - H(X)} \leq \eps_N,
\end{equation}
\item
\begin{equation}
\label{eq:xl_full_prob}
\sum_{x^l \in G} P_X(x^l) \geq 1 - \eps_N, 
\end{equation}
\end{enumerate}
where $\eps_N \rightarrow 0$ with $\tilde{w}N \rightarrow \infty$. 
\end{lem}
This Lemma is a version of an Asymptotic Equipartition Property (AEP) for independent variables (see 
\cite{coverthomas}). Statements (\ref{eq:xl_close_to_exp}) and (\ref{eq:xl_full_prob}) follow from a weak 
law of large numbers. The details of the proof are identical to the standard AEP and are omitted.

Note that on one hand sequences $x^l$ are a set of almost full 
probability by (\ref{eq:xl_full_prob}), and on the other hand, by 
Lemma \ref{lem:empirical_mx_bound}, $\Probu{X}{M(x) \in U}$ is 
also close to $1$. Denote 
\begin{equation}
X^l = G \cap O.
\end{equation}
It follows that 
\begin{equation}
\label{eq:xl_xl_prob}
\Probu{X}{X^l} \geq 1 - \eps_N,
\end{equation}
perhaps with a slightly different $\eps_N$. In addition, similarly to the standard AEP, we can obtain 
cardinality estimates on $\Abs{X^l}$.
Indeed, combining (\ref{eq:xl_xl_prob}) and (\ref{eq:xl_close_to_exp}) we get 
\begin{equation}
\label{eq:xl_cardinality}
2^{N (H(X) + \eps_N)} \geq  \Abs{X^l} \geq 2^{N (H(X) - \eps_N)}.
\end{equation}

Next, using  (\ref{eq:entropy_hmm_bound}) we 
can write  
\begin{equation}
 \sum_{x^l \in X^l} \Probu{H}{x^l} \leq \Probu{H}{ O} \leq 2^{-N D^2},
\end{equation}
or equivalently, 
\begin{equation}
\label{eq:average_statement_intermediate}
\frac{1}{\Abs{X^l}} \sum_{l} \Probu{H}{x^l} \leq 
2^{-N \Brack{D^2 + \frac{1}{N}\log |X^l|}}. 
\end{equation}
Using (\ref{eq:xl_cardinality}) we therefore obtain 
\begin{equation}
\label{eq:average_statement}
\frac{1}{\Abs{X^l}} \sum_{l} \Probu{H}{x^l} \leq 
2^{-N \Brack{D^2 + \sum_{i \leq k} w_i H(\mu_i) - \eps_N }}. 
\end{equation}

Note that (\ref{eq:average_statement}) is essentially the statement 
of Theorem \ref{thm:main_thm} on average over $x^l$. 
By applying  Markov inequality to this average we get that the proportion of 
$x \in X_l$ which satisfy 
\begin{equation}
\label{eq:x_l_likelihood}
L(x,H) \geq -D^2/2 - \sum_j w_j H(\mu_j) - \eps_N
\end{equation}
is at most $2^{-N \frac{D^2}{2}}$. Formally, 
\begin{equation}
\label{eq:x_l_proportion_statement}
\frac{\Abs{\Set{x \in X^l \setsep x \mbox{ satisfies (\ref{eq:x_l_likelihood})}}} }{ \Abs{X^l} } \leq 
2^{-N \frac{D^2}{2}}.
\end{equation}
Using (\ref{eq:xl_close_to_exp}) and (\ref{eq:xl_cardinality}) again,
this estimate implies a probability estimate over $X$, 
\begin{align}
\label{eq:proof_final_markov}
&\Probu{X}{ L(x,H) \geq -D^2/2 - \sum_j w_j H(\mu_j) + \eps_N} \\
&\leq 2^{-N \Brack{\frac{D^2}{2} - 2\eps_N}} + \eps_N, \nonumber
\end{align}
therefore concluding the proof of Theorem 
\ref{thm:main_thm}. 

It remains to prove the bound (\ref{eq:entropy_hmm_bound}).  We begin with a standard construction 
for transforming an HMM into a Markov chain in a special form. This  
converts the problem of bounding the likelihood of data under an HMM to a problem 
of bounding a 
likelihood of a certain set of paths in the chain.  Given an 
HMM $H = H(S,\{\nu_i)\}_{1}^k, \{p_{ij}\}_{i,j=1}^k)$,
construct a Markov chain $H' = (S',p')$ with state space $S' = 
S\times 
\groundX$, 
and transition probabilities 
\begin{equation}
	p'_{(i,a),(j,b)} = p_{ij} \nu_j(b).
\end{equation}
For a state $(i,a) \in S'$, we refer to $a$ as the data component of the 
state. Clearly, by observing a random walk of $H'$ and looking 
only at the data component, we get a distribution 
over the data that is identical to that of the HMM. Note that for a 
single data vector $x = (x_1,\ldots,x_{N+1})$, there are exactly 
$k^{N+1}$ paths of the chain $H'$ yielding the data $x$. 

Next, we use type theory 
for Markov chains to obtain deviation bounds on the empirical second 
moment of a random walk. Similarly to second moment of the data, for 
a Markov chain $H' = (S',p')$, and a path $s = s_1,s_2,
\ldots,s_{N+1}$, where $s_i \in S'$, define the second moment $M(s) \in \Delta_{S' \times S'}$ by
\begin{align}
&M(s)(u,v) = \\
&\frac{1}{N} \Abs{
\Set{ i \leq N \setsep s_{i} = u \wedge s_{i+1} = v} 
},
\end{align}
for all $u,v \in S'$. For a subset $\Pi \subset \Delta_{S' \times S'}$, the second order type theory provides bounds of the form 
\begin{equation}
\label{eq:type_demo}
\Probu{H'}{ M(s) \in \Pi } \leq 2^{-N\cdot D}, 
\end{equation}
where $D$ is a suitably defined distance between the set $\Pi$ and 
the transition matrix $p'$. Statement (\ref{eq:type_demo}) is a 
Markov chain analog of Sanov's theorem for i.i.d sequences 
(\cite{sanov}, \cite{coverthomas}). We use a second moment deviation 
inequality due to \cite{ccc87}, stated as Lemma 
\ref{lem:sanov}. Note that type theory provides estimates on moments 
of paths of the chain $H'$, which take values in $\Delta_{S' \times S'}$, while our  
assumptions are about moments of the data, $M(x) \in 
\Delta_{\groundX \times \groundX}$. We now describe the connection between the two types of moments. 
Consider a Markov chain $H'=(S',p')$ corresponding to an HMM $H$. Define 
a linear map $T : \Delta_{S' \times S'} \rightarrow \Delta_{\groundX 
\times \groundX}$ by 
\begin{equation}
T(M')(a,b) = \sum_{i,j \leq k} M'((i,a),(j,b)).
\end{equation}
If $M'$ is the second moment of a path of the chain, then $T(M')$
is the second moment of the data.  The map $T$ satisfies the following inequality, which is crucial 
for our analysis. 
\begin{lem}
For any $M_1,M_2 \in \Delta_{S' \times S'}$, 
\begin{equation}
\label{eq:T_contraction}
  D(T(M_1)|T(M_2)) \leq D(M_1 | M_2). 
\end{equation}
\end{lem}
\begin{proof} 
This result is a consequence of the chain rule for relative entropies. 
To see this, 
represent an element $v = ((i,a),(j,b)) \in S' \times S'$ as a pair $v = (u,w)$ where $u = (i,j)$ and 
$w=(a,b)$ are the state and data parts of $v$. Then $M \in \Delta_{S' \times S'}$ is a distribution 
over all $(u,w)$. Denote by $V = (U,W)$ the random vector with values in $S' \times S'$ and distribution 
$M$. Then, by definition, $T(M)$ is the marginal distribution of component $W$ of $V$. 
By the chain rule for relative entropies (\cite{coverthomas}, equation (2.67) ),
\begin{align}
\label{eq:T_contraction_proof}
&D(M_1,M_2) = \\
&D( (U_1,W_1) | (U_2,W_2)) = \nonumber \\
& D( W_1 | W_2) +  \nonumber \\
& \spaceo +\sum_{a,b \in \groundX} M_1(a,b) \cdot D\Brack{ \BrackSq{ U_1 | W_1 = (a,b)} | \BrackSq{ U_2 | W_2 = (a,b)} } \nonumber \\
& \geq D( W_1 | W_2) \nonumber \\ 
& = D(T(M_1)|T(M_2)). \nonumber  
\end{align}
where the inequality in (\ref{eq:T_contraction_proof}) is due to the non-negativity of relative entropy. 
\end{proof}

To state the deviation result, Lemma \ref{lem:sanov}, we require some 
additional notation. For any measure $M \in \sqmeas$, define the left and right marginalizations 
$\bar{M}, \dbar{M} \in \sinmeas$  by 
\begin{equation}
\bar{M}(u) = \sum_{v\in S'} M(u,v) ,\spaceo  
\dbar{M}(u) = \sum_{v\in S'} M(v,u).
\end{equation}
Moreover, given $M \in \sqmeas$, define the related transition 
matrix to be 
\begin{equation}
M(v|u) = \frac{M(u,v)}{\bar{M}(u)}.
\end{equation}
A measure $M \in \sqmeas$ is called stationary, if $\bar{M} = 
\dbar{M}$. Such measure is a stationary measure of a random walk 
given by the transition matrix $M(v|u)$. We denote by $\sqmeas^0$ 
the set of all stationary measures. 

Finally, we introduce a quantity that will control the deviations of 
moments. Given a transition matrix $p'=p'_{uv}$ and a measure 
$M \in \sqmeas$, define 
\begin{eqnarray}
\label{eq:dmp_def_1}
D(M | p' ) = & \sum_{u,v \in S'} M(u,v) \log \frac{M(v|u)}{p'_{uv}} = \\
\label{eq:dmp_def_2}
&\sum_{u,v \in S'} M(u,v) \log \frac{M(u,v)}{\bar{M}(u)p'_{uv}}.
\end{eqnarray}

The quantity $D(M|p')$ differs from the standard Kullback-
Leibler divergence since $p'$ is not a measure. However, as follows 
from (\ref{eq:dmp_def_2}), we can write
$D(M|p') = D(M|z)$, where $D(t|z)$ is the standard KL divergence and 
$z\in \sqmeas$ is defined by $z(u,v) = \bar{M}(u) \cdot p'_{uv}$.

For any closed set $\Pi \subset \sqmeas$, denote $\Pi_0 = \sqmeas^0 \cap \Pi$. 

\begin{lem}[\cite{ccc87}]
\label{lem:sanov}
Let  $\Pi \subset \sqmeas$ be  a closed convex set. For any $D' > 0$ there is a sequence 
$\eps_N$ with $\lim_{N \rightarrow \infty} \eps_N = 0$ such that for any Markov chain 
 $C = (S',p')$ satisfying (\ref{eq:ccc_condition}),
\begin{equation}
\label{eq:ccc_condition}
D'  = \min_{M \in \Pi_0} D(M|p'),
\end{equation}  
if $X = X_1,\ldots,X_{N+1}$ is a random walk generated by $C$ then 
\begin{equation}
\Probu{C}{ M(X) \in \Pi} \leq 2^{-N \Brack{D' - \eps_N}}.
\end{equation}
\end{lem}

Lemma \ref{lem:sanov} provides us with likelihood estimates 
that depend on the parameters of the unfolded Markov chain $H'$. 
To obtain the bound (\ref{eq:entropy_hmm_bound}) for an HMM $H$ we 
apply Lemma \ref{lem:sanov} to the Markov chain $H'$ with the  set 
$\Pi \subset \sqmeas$ given by
\begin{equation}
\Pi = T^{-1}(U) = \Set{M' \setsep T(M') \in U},
\end{equation}
where $U_{\eps}$ was defined in (\ref{eq:u_eps_def}). 

Choose some $M' \in \Pi$ and denote $M = T(M')$.
In what follows we show that 
if $D$ is given by (\ref{eq:m_x_m_h_dist}), then 
\begin{equation}
\label{eq:deviation_application_base}
 D(T(M')|T(\bar{M'} p' )) \geq 2 D^2. 
\end{equation} 
Note that by (\ref{eq:T_contraction}) this implies 
\begin{equation}
 D(M'|\bar{M'} p' ) \geq 2 D^2, 
\end{equation} 
and hence $D' \geq 2 D^2$ in Lemma \ref{lem:sanov}, therefore proving 
(\ref{eq:entropy_hmm_bound}). 

Next, to obtain (\ref{eq:deviation_application_base}), observe that 
by Pinsker's Inequality (see \cite{coverthomas} 
\footnote{Pinsker's Inequality: 
$2\norm{\mu - \nu}_{TV}^2 \leq D(\mu|\nu)$ for all measures $\mu,\nu$.
} 
), it is 
sufficient to show that 
\begin{equation}
\label{eq:deviation_application_tv}
  \norm{T(M') - T(\bar{M'} p') }_{TV}  \geq D. 
\end{equation} 
Recall that by definition $T(M') \in U$, and hence 
\begin{equation}
  \norm{T(M') - M_X }_{TV}  \leq \frac{3}{m}. 
\end{equation} 
Thus to obtain (\ref{eq:deviation_application_tv}) it is sufficient 
to show that 
\begin{equation}
\label{eq:deviation_application_tv2}
  \norm{M_X - T(\bar{M'} p') }_{TV}  \geq D + \frac{3}{m}. 
\end{equation} 
Let us now write the explicit expression for $T(\bar{M'} p')$. 
\begin{eqnarray}
T(\bar{M'} p' )(a,b) = 
\sum_{i,j \leq k} \bar{M'}((i,a)) p'_{(i,a),(j,b)} = \\
\sum_{i,j \leq k} \bar{M'}((i,a)) p_{ij} \mu_j(b). \label{eq:tmfull}
\end{eqnarray}
In addition, observe that by definition, with the notation $M = T(M')$,
\begin{eqnarray}
\sum_{i \leq k} \bar{M'}((i,a)) = \\
\sum_{i \leq k} \sum_{j \leq k} \sum_{b \in \groundX} M'((i,a),(j,b)) 
= \\
\sum_{b \in \groundX} M(a,b) = \bar{M}(a). \label{eq:mprimembarmarg}
\end{eqnarray}
For every $a,b \in \groundX$, denote  $\phi_a = (M'(1,a),\ldots,M'(k,a)) \in \RR^k$, and $\chi_b = (\nu_1(b),
\ldots,\nu_k(b))$. Then we can rewrite (\ref{eq:tmfull}) in 
the generalized second moment form (\ref{eq:generalized_moment_def}) as
\begin{equation}
T(\bar{M'} p' )(a,b) = \phi_a \cdot p \cdot \chi_b.
\end{equation}
Moreover, since $M = T(M') \in U$, the marginals satisfy 
\begin{equation}
\label{eq:marginals_ph_a_condition}
\norm{\bar{M_X} - \bar{M}}_{TV} \leq \norm{M_X - M}_{TV} \leq \frac{3}{m}.
\end{equation}
Therefore, by (\ref{eq:mprimembarmarg}) and (\ref{eq:marginals_ph_a_condition})  the condition (\ref{eq:main_condition_on_phi_a}) for $\phi_a$ in the main Theorem  
holds.

\subsection{Proof of Theorem \ref{thm:full_statement}}
\label{sec:proof_of_full_statement}
\newcommand{\Hdelta}{\mathcal{H}_{\delta}}
\begin{proof}[Proof of Theorem \ref{thm:full_statement}]
Set $D_0 = \sqrt{\frac{\log 3km}{m}}$.
To obtain Theorem \ref{thm:full_statement} it suffices to show that 
with high probability over $X$, 
\begin{align}
\label{eq:full_statement_proof_likelihood}
L(x,H,\pi) \leq & -D_0^2 - \sum_{j} w_j H(\mu_j) + \eps(\Nmin) \\
= &	-\frac{\log 3km}{m} - \sum_{j} w_j H(\mu_j) + \eps(\Nmin) \nonumber
\end{align}
jointly for all HMMs $H \in \mathcal{H}_{\delta}$ which satisfy 
\begin{equation}
\label{eq:full_statement_big_d}
D(H) > \sqrt{\frac{\log 3km}{m}}.
\end{equation}
Indeed, assume that (\ref{eq:full_statement_proof_likelihood}) holds for all $H$ which satisfy
 (\ref{eq:full_statement_big_d}). Then, since by Lemma \ref{lem:right_path_full} we know that there exists an HMM $H_0 \in \mathcal{H}_{\delta}$ such that 
\begin{equation}
\label{eq:full_statement_proof_H_0_likelihood}
L(x,H_0,\pi) > -\frac{\log 3km}{m} - \sum_{j} w_j H(\mu_j) - \eps(\Nmin),
\end{equation}
it follows that the maximum likelihood estimator $H$ must satisfy
\begin{equation}
D(H) \leq \sqrt{\frac{\log 3km}{m}}.
\end{equation}
Note that for a single fixed HMM $H$ satisfying (\ref{eq:full_statement_big_d}), the statement 
(\ref{eq:full_statement_proof_likelihood}) holds by Theorem \ref{thm:main_thm} with high probability.
However, since we would like to have explicit exponential probability bounds, we work 
directly with estimate \ref{eq:x_l_proportion_statement} in the proof of Theorem \ref{thm:main_thm} rather 
than with the final statement of that theorem. Then, for a fixed $H$, the probability over $X^l$ of 
$x$ not satisfying (\ref{eq:full_statement_proof_likelihood}) is at most $2^{-N \frac{D_0^2}{2}}$.
The uniform statement for all $H \in \mathcal{H}_{\delta}$ satisfying 
(\ref{eq:full_statement_big_d}) can be obtained by approximation and union bound. We first define an appropriate metric on $\Hdelta$.
Consider the set $\mathcal{H}_{\delta}$ as a subset of the Euclidean space $\RR^F$, 
where $F = k^2 + k\cdot \Abs{\groundX}$ and we simply consider the parameters of an HMM as coordinates. 
For any $H \in \mathcal{H}_{\delta}$ let $v(H) = (v_t(H))_{t=1}^F \in \RR^F$ be the vector corresponding to $H$. In what follows we identify $\Hdelta$ with a subset of $\RR^F$, $\Set{v(H) \setsep H \in \Hdelta} \subset \RR^F$.
By definition we have for every $H \in \Hdelta$,
\begin{equation}
\label{eq:full_statement_v_t_delta}
 v_t(H) \geq \delta \spaceo \forall t \leq F.
\end{equation}
Define a map $R : \RR^F \mapsto \RR^F$ by 
\begin{equation}
R(v) = (\log v_1, \ldots, \log v_F)
\end{equation}
and define a metric on $\Hdelta$ by 
\begin{align}
d_{*}(H_1,H_2) =& \norm{R(v(H_1)) - R(v(H_2))}_{\infty} \\
 =& \max_{t\leq F} \Abs{\log \frac{v_t(H_1)}{v_t(H_2)} }. \nonumber
\end{align}
Next, for $\gamma >0$ let $\Gamma_{\gamma}$ be the minimal cardinality of a $\gamma$-net of $\Hdelta$ with 
respect to the metric $d_{*}$. Since $\Hdelta$ is bounded in $\RR^F$, the map $R$ is coordinate-wise at 
most $\frac{1}{\delta}$-Lipschitz on $\Hdelta$ (by (\ref{eq:full_statement_v_t_delta})), and $d_{*}$ is an $\ell_{\infty}$ norm on the image of $R$,
standard volumetric arguments imply that 
\begin{equation}
\Gamma_{\gamma} \leq c \Brack{\frac{1}{\gamma \cdot \delta} }^F. 
\end{equation}
It is also easy to check that the normalized log-likelihood is $1$-Lipschitz with respect 
to $d_{*}$:
\begin{equation}
\Abs{L(x,H_1,\pi) - L(x,H_2,\pi)} \leq d_{*}(H_1,H_2),
\end{equation}
for every sequence $x$ and distribution $\pi$. 

Consider an $1/N$-net in $\Hdelta$. 
As noted above, for an individual HMM $H$ satisfying (\ref{eq:full_statement_big_d}), the statement 
(\ref{eq:full_statement_proof_likelihood}) holds with probability at least $1 - 2^{-N \frac{D_0^2}{2}}$
over $X^l$. 
Therefore with probability at least 
\begin{equation}
\label{eq:full_statement_prob_final}
1 - 2^{-N \frac{D_0^2}{2}} \cdot 2^ {F \log \frac{1}{N \cdot \delta}},
\end{equation}
we have 
\begin{equation}
\label{eq:full_statement_likelihood_after_lip}
L(x,H,\pi) \leq -D_0^2 - \sum_{j} w_j H(\mu_j) + \eps(\Nmin) + \frac{1}{N}. 
\end{equation}
It remains to observe that for any $N$ such that 
\begin{equation}
N \geq  2\frac{F \log \frac{1}{N \cdot \delta}}{D_0^2},
\end{equation}
the probability in (\ref{eq:full_statement_prob_final}) is positive, and approaches 1 for larger $N$, thereby completing the proof. 
\end{proof}

\bibliographystyle{apalike}
\bibliography{hmm_stability}

\end{document}